\documentclass{ifacconf}

\usepackage{graphicx}      % include this line if your document contains figures
\usepackage[numbers]{natbib}        % required for bibliography

\usepackage{multicol}
\usepackage{amsmath} % assumes amsmath package installed
\usepackage{amssymb}  % assumes amsmath package installed
\usepackage{mathtools}
\usepackage{tablefootnote}
\usepackage{optidef}
\usepackage{caption}
\usepackage{subcaption}
\usepackage{float}
\usepackage{threeparttable}
\usepackage{tablefootnote}

\usepackage{siunitx}

\usepackage{epstopdf}
\usepackage{multirow}
\usepackage[leftcaption]{sidecap}
\usepackage{booktabs}
\usepackage{arydshln}
\usepackage{lipsum}
\usepackage{comment}
\usepackage[dvipsnames]{xcolor}
\usepackage{enumitem}
\usepackage[nolist]{acronym} % acronyms by the \ac{label} command

\newacro{vic}[VIC]{Variable Impedance Control}
\newacro{nn}[NN]{Neural Network}
\newacro{dnn}[DNN]{Deep Neural Network}
\newacro{cnn}[CNN]{Convolutional Neural Network}
\newacro{penn}[PE]{Probabilistic ensemble NN}
\newacro{mpc}[MPC]{Model Predictive Control}
\newacro{mpic}[MPIC]{Model Predictive Impedance Control}
\newacro{pi2}[PI\textsuperscript{2}]{Policy Improvement with Path Integrals}
\newacro{dmp}[DMP]{Dynamic Movement Primitive}
\newacro{seds}[SEDS]{Stable Estimator of Dynamical Systems}
\newacro{ilc}[ILC]{Iterative  learning  control}
\newacro{gmm}[GMM]{Gaussian Mixture Model}
\newacro{gmr}[GMR]{Gaussian Mixture Regression}
\newacro{gp}[GP]{Gaussian Processes}
\newacro{ppc}[PPC]{Passivity-Preservation Control}
\newacro{tpgmm}[TP-GMM]{Task-Parameterized \ac{gmm}}
\newacro{lfd}[LfD]{Learning from Demonstration}
\newacro{lfhd}[LfD]{Learning-from-human-Demonstration}
\newacro{il}[IL]{Imitation Learning}
\newacro{wls}[WLS]{weighted least-squares}
\newacro{spd}[SPD]{Symmetric Positive Definite}
\newacro{emg}[EMG]{Electromyography}
\newacro{vil}[VIL]{Variable Impedance Learning}
\newacro{vilc}[VILC]{Variable Impedance Learning Control}
\newacro{ai}[AI]{Artificial Intelligent}
\newacro{gadmp}[Ga-DMP]{Geometry-aware Dynamic Movement Primitives}
\newacro{promp}[ProMP]{Probabilistic Movement Primitives}
\newacro{kmp}[KMP]{Kernelized Movement Primitives}
\newacro{dof}[DoF]{Degree of Freedom}
\newacro{msd}[MSD]{virtual-Mass Spring-Damper}
\newacro{gadmp}[GaDMP]{Geometry-aware \ac{dmp}}
\newacro{mbrl}[MBRL]{Model-based Reinforcement Learning}
\newacro{ic}[IC]{Impedance Control}
\newacro{rl}[RL]{Reinforcement Learning}
\newacro{hrc}[HRC]{Human-Robot Collaboration}
\newacro{hri}[HRC]{Human-Robot Interaction}
\newacro{pets}[PETS]{probabilistic ensembles with trajectory sampling}
\newacro{cem}[CEM]{Cross Entropy Method}
\newacro{dcem}[DCEM]{Differentiable Cross Entropy Method}
\newacro{mdp}[MDP]{Markov Decision Process}
\newacro{dpg}[DPG]{Deterministic Policy Gradient}
\newacro{spg}[SPG]{Stochastic Policy Gradient}
\newacro{td}[TD]{Temporal Difference}
\newacro{ocp}[OCP]{Optimal Control Problem}
\newacro{ac}[AC]{Actor-Critic}
\newacro{lstd}[LSTD]{least-squares temporal-difference learning algorithm}
\newacro{ift}[IFT]{Inverse Function Theorem}
\newacro{kkt}[KKT]{Karush-Kuhn-Tucker}
\newacro{ppo}[PPO]{Proximal Policy Optimization}

\newcommand{\figsref}[1]{Figures~\hyperref[#1]{\ref*{#1}}}
\newcommand{\Figref}[1]{Figure~\hyperref[#1]{\ref*{#1}}}

\newcommand{\tabref}[1]{Tab.~\hyperref[#1]{\ref*{#1}}}

\newcommand{\algoref}[1]{Algorithm~\hyperref[#1]{\ref*{#1}}}

\newlength{\Oldarrayrulewidth}

\definecolor{darkgreen}{rgb}{0.0,0.49,0.19}

\graphicspath{{figures/}}

\newtheorem{theorem}{Theorem}

\newif\ifdraft
\draftfalse % or 

\newcommand{\vect}[1]{\ensuremath{\boldsymbol{\mathrm{#1}}}}

\begin{document}
\begin{frontmatter}

%\title{MDPs with Future Information: How to Deliver Optimal Policies with MPC?} 
%\title{Solving Markov Decision Processes with Future Information via Model Predictive Control}
\title{Solving Markov Decision Processes with Future Information via MPC}

\thanks[footnoteinfo]{This work was supported by the Research Council of Norway through the project \emph{Safe Reinforcement Learning using MPC}.}
\thanks{© 2026 the authors. This work has been accepted to IFAC for publication under a Creative Commons Licence CC-BY-NC-ND.}

\author[First]{Shambhuraj Sawant} 
\author[First]{Akhil S Anand} 
\author[First]{Dirk Reinhardt}
\author[First]{Sebastien Gros}

\address[First]{Norwegian University of Science and Technology (NTNU), Trondheim, Norway. E-mail:{\tt\small \{shambhuraj.sawant, akhil.s.anand, dirk.p.reinhardt, sebastien.gros\}@ntnu.no}}

\begin{abstract}                % Abstract of 50--100 words
Model Predictive Control (MPC) is widely used in industrial and robotic systems for enforcing constraints and embedding domain knowledge through finite-horizon optimization-based planning. 
However, despite these strengths, an MPC scheme typically does not yield optimal policies for sequential decision-making problems formulated as Markov Decision Processes (MDPs). 
Recent combinations of MPC with Reinforcement Learning (RL) alleviate this issue by treating MPC as a parameterized model of the optimal policy of an MDP and adjusting its parameters using data.
While these approaches typically consider classical MDPs, many real-world problems include future information—such as forecasts, prices, or reference trajectories—at decision time, which must be included in the MDP state for optimal decision-making.
Current MPC-RL approaches do not directly account for this augmented-state structure, raising the question of how to incorporate future information into MPC to obtain an optimal policy.
This work establishes the structural requirements under which a parameterized MPC can exactly represent the optimal value functions and policy of an MDP with future information. 
We further demonstrate that such a parameterized MPC can serve as a structured function approximator, with its parameters learned using RL. 
The approach is illustrated on a point-mass racing task with future reference information.

% Model Predictive Control (MPC) is widely used in industrial and robotic systems due to its ability to enforce constraints and embed domain knowledge through finite-horizon optimization-based planning. 
% However, despite these strengths, an MPC scheme typically does not yield optimal policies for sequential decision-making problems formulated as Markov Decision Processes (MDPs). 
% Recent combinations of MPC with Reinforcement Learning (RL) help alleviate this issue by treating MPC as a parameterized model of the optimal policy of an MDP and adjusting its parameters using data.
% While these approaches typically consider the classical MDP setting, many real-world problems include relevant future information—such as forecasts, future prices, or reference trajectories—at decision time, which must be included in the MDP state for optimal decision making.
% Current MPC-RL approaches do not directly account for this augmented state structure, raising the question of how to incorporate future information into an MPC scheme to obtain an optimal policy.
% This work establishes the structural requirements under which a parameterized MPC can exactly represent the optimal value functions and policy of an MDP with future information. 
% We further demonstrate that a parameterized MPC scheme can serve as a structured function approximator, with its parameters learned using RL. 
% The approach is illustrated on a point-mass racing task with future reference information.
\end{abstract}

\begin{keyword}
% Five to ten keywords, preferably chosen from the IFAC keyword list.
Markov Decision Process, Model Predictive Control, Reinforcement Learning.
\end{keyword}

\end{frontmatter}
%===============================================================================

\section{Introduction} \label{sec:intro}
\ac{mpc} is a widely adopted methodology for sequential decision-making in industrial and robotic systems, as it can satisfy constraints and performance objectives through finite-horizon optimization~\citep{rawlings2017model}.
Its strength lies in repeatedly solving an \ac{ocp} in a receding-horizon fashion, enabling constraint enforcement and the systematic incorporation of domain knowledge.

An \ac{mpc}, when viewed as a solution method for a \ac{mdp} modeling the sequential decision-making problem, typically provides only an approximate solution.
% However, when viewed as a solution method for a \ac{mdp} modeling the underlying sequential decision-making problem, an \ac{mpc} typically provides only an \emph{approximate} solution. 
%real-world sequential decision-making problems, typically modeled as an \ac{mdp}, a standard \ac{mpc} scheme provides only an \emph{approximate} solution.
Despite strong empirical performance, an \ac{mpc} rarely yields \emph{optimal} decisions, as it relies on several simplifications to keep the optimization problem tractable.

Learning-based \ac{mpc} methods attempt to mitigate these limitations by adapting or correcting \ac{mpc} formulations using data.
Approaches such as Gaussian-process \ac{mpc}~\citep{hewing2020gpmpc} and Bayesian \ac{mpc}~\citep{wabersich2020bmpc} focus on improving predictive accuracy or quantifying model uncertainty, with broader surveys in~\citep{hewing2020learning, mesbah2022fusion}.
Alternatively, an \ac{mpc} can be viewed as a model of an \ac{mdp} which can deliver optimal value functions and policy even with an inaccurate dynamics model \citep{gros2019data, reinhardt2025economic}.
This perspective motivates the use of a parameterized \ac{mpc} as a structured function approximator to solve \acp{mdp}. 
Recent approaches combine the structured function-approximation capabilities of \ac{mpc} with \ac{rl} methods to improve closed-loop performance, viz. Actor–Critic \ac{mpc}~\citep{romero2024actor}, and an overview of \ac{mpc}-\ac{rl} combination approaches appears in~\citep{reiter2025synthesismodelpredictivecontrol}.

The existing results in \citep{gros2019data, reinhardt2025economic} assume the classical \ac{mdp} setting, where decisions depend solely on the current physical state of the system. 
In contrast, many real-world problems include future information at decision time—such as weather-dependent renewable energy forecasts, fluctuating electricity prices in smart grids, or reference trajectories for autonomous racing—that is relevant to optimal decision-making.
In such problems, the available future information must be incorporated into the state definition to satisfy the \emph{Markov} property~\citep{powell2019rlso, de2021incorporating}.
Existing results in~\citep{gros2019data} do not trivially extend to such \acp{mdp} with future information. 
In contrast, classical \ac{mpc} extensions incorporate future information in a stage-wise or heuristic manner, without explicitly accounting for the augmented state required by the underlying \ac{mdp}, and may fail to optimally exploit the available information.
This raises the key question of how an \ac{mpc} should be formulated to represent the optimal policy of an \ac{mdp} with future information.
% This mismatch raises a key question: how should an \ac{mpc} be formulated to represent the optimal policy of an \ac{mdp} with future information?

In this work, we address the structure and parameterization required for an \ac{mpc} to deliver an optimal policy for an \ac{mdp} with future information. % and further provide an analysis of the implications for learning the \ac{mpc} parameterization using \ac{rl} techniques.
% This paper addresses this gap by extending results in \citep{gros2019data} to \ac{mdp}s with future information and analyzing the implications for \ac{mpc} formulation and for learning the \ac{mpc} parameterization using \ac{rl} techniques.
Our contributions are:
\begin{itemize}
\item We establish conditions under which an \ac{mpc} can represent the optimal value functions and optimal policy for an \ac{mdp} with future information.
\item We propose a practical parameterized \ac{mpc} formulation that incorporates future information and can be combined with \ac{rl} to learn its parameters.
\item We demonstrate the proposed formulation on a point-mass racing task with future reference information using \ac{ppo}~\citep{schulman2017proximal} to learn the \ac{mpc} parameters.
% We demonstrate the proposed formulation on a point-mass racing task with future reference information, where the \ac{mpc} parameters are learned using \ac{ppo}~\citep{schulman2017proximal}.
% We demonstrate the proposed formulation on a point-mass racing task with future reference information, where the proposed \ac{mpc} learned using \ac{ppo}~\citep{schulman2017proximal} exploits forecasts more effectively than both classical \ac{mpc} and standard parameterized \ac{mpc} scheme learned using RL.
\end{itemize}

The article is organized as follows.
Section~\ref{sec:background} reviews \acp{mdp}, \acp{mpc}, and their treatment of future information.
Section~\ref{sec:mpc_optimal_policy} formulates the augmented-state \ac{mpc} scheme and establishes conditions under which it can represent the optimal value functions and policy of an \ac{mdp} with future information.
Section~\ref{sec:mpc_fn_approx} presents a practical parameterized formulation, discusses its use within an \ac{rl} pipeline, and illustrates the approach on a point-mass racing task.
Section~\ref{sec:conclusion} concludes the paper.
% The article is organized as follows.
% Section \ref{sec:background} provides background on \acp{mdp} and \acp{mpc}, and how each framework handles future information.
% Section \ref{sec:mpc_optimal_policy} introduces our proposed \ac{mpc} formulation for solving \acp{mdp} with future information and presents the main theoretical result: the conditions under which an \ac{mpc} yields the optimal value functions and policy for an \ac{mdp} when future information is available.
% Section \ref{sec:mpc_fn_approx} then discusses how the proposed parameterized \ac{mpc} can be formulated in practice and combined with \ac{rl} methods, and includes an illustrative example, with Section \ref{sec:conclusion} providing the conclusion. 

\section{Background}\label{sec:background}
In this section, we provide an overview of \acp{mdp}, \acp{mpc}, outline how an \ac{mpc} can approximate optimal policies, and then discuss their treatment of future information.

\subsection{Markov Decision Processes}
An \ac{mdp} provides a mathematical framework for modeling discrete-time sequential decision-making under uncertainty~\citep{puterman2014markov}.
It is defined over a state space $\mathcal{S}$ and an action space $\mathcal{A}$, with a transition kernel $\mathcal{P}$ giving the conditional probability (or density) of an action $\vect a\in\mathcal{A}$ taken in a state $\vect s\in \mathcal{S}$ leading to a next state $\vect s_+\in\mathcal{S}$, 
\begin{equation}
	\mathcal{P}[\,\vect s_+\,|\,\vect s,\vect a\,], \label{eq:tf1}
\end{equation}
and a stage cost, $L$.
Here, the state $\vect s$ is \emph{Markovian}, i.e., state $\vect s$ together with the transition kernel $\mathcal{P}$ provides complete statistics for the distribution of the successor-state $\vect s_+$ for any action $\vect a$.
%Note that, typically, \eqref{eq:tf1} is given as $\vect s_+= \vect f(\vect s, \vect a, \vect w)$ in the control literature with a possibly nonlinear function $\vect f$ and a random disturbance $\vect w$ drawn from a distribution $\mathcal{W}$. At the same time, the cost $L$ is analogous to the reward function in RL literature.
A policy function $\vect \pi$ then gives an action $\vect a$ to execute in a state $\vect s$, i.e. $\vect \pi: \mathcal{S}\rightarrow\mathcal{A}$. 
% It can be deterministic, mapping a state $\vect s$ to an action $\vect a$, i.e. $\vect \pi: \mathcal{S}\rightarrow\mathcal{A}$, or stochastic, specifying a conditional distribution over actions given a state $\vect s$, i.e. $\vect\pi: \mathcal S \rightarrow \mathcal P(\mathcal A)$.

Solving an \ac{mdp} consists of finding an \emph{optimal} policy that minimizes the objective $J$, typically given as,
\begin{equation}
J(\vect \pi) =\ \mathbb E_{\tau \sim (\vect \pi, \mathcal{P})}\left [\left.\sum_{t=0}^\infty \gamma^t L(\vect s_t,\vect a_t)\,\right|\,\vect s_0 \sim \rho_0 \,\right],\label{eq:j1}
\end{equation}
where, $\tau=\{\vect s_0, \vect a_0, \vect s_1, \vect a_1, ...\}$ denotes the trajectory generated by a policy $\vect \pi$ under $\mathcal{P}$, $\rho_0$ is the given initial state distribution, $\gamma \in (0, 1]$ is the discount factor, and $t$ denotes the time index.
% In \eqref{eq:j1}, $\mathbb E_{\tau \sim (\vect \pi, \mathcal{P})}[.]$ denotes expectation taken over the Markov chains resulting from following a policy $\vect \pi$ under the transition dynamics $\mathcal{P}$.
Then the optimal policy is, 
\begin{equation}
    \vect \pi^\star = \arg\min_{\vect \pi} J(\vect \pi).\label{eq:optpi}
\end{equation}
% In this work, we assume that such an \ac{mdp} yields a well-posed problem, i.e., it has an optimal policy defined over a part of its state space $\mathcal{S}$.

Finding an optimal policy typically involves computing the value functions, namely the state value function $V^{\vect \pi}$ and the state–action value function $Q^{\vect \pi}$, associated with a policy $\vect \pi$~\citep{sutton2018reinforcement}.
Here, $V^{\vect \pi}$ denotes the expected cumulative cost when starting from a state $\vect s$ and following a policy $\vect \pi$, and $Q^{\vect \pi}$ denotes the expected cumulative cost when taking an action $\vect a$ in a state $\vect s$ and thereafter following $\vect \pi$,
\begin{align}
	V^{\vect \pi}(\vect s) &= \mathbb E_{\tau \sim (\vect \pi, \mathcal{P})}\left [\left.\sum_{t=0}^\infty \gamma^t L(\vect s_t,\vect a_t)\,\right|\,\vect s_0 = \vect s \,\right], \\
	Q^{\vect \pi}(\vect s, \vect a) &= L(\vect s, \vect a) + \gamma \mathbb{E}_ {\mathcal{P}} \left[\left.V^{\vect \pi}(\vect s_+)\,\right|\,\vect s, \vect a \right].
\end{align} 
We label $Q^\star$ as the optimal state-action value function and $V^\star$ as the optimal state value function associated with $\vect \pi^\star$. The optimal value functions and policy are related through Bellman equations~\citep{sutton2018reinforcement}.
% \begin{align}
%     Q^\star (\vect s, \vect a) &= L(\vect s, \vect a)+\gamma \mathbb{E}_ {\mathcal{P}}\left[\left. V^\star(\vect s_+) \right| \vect s, \vect a \right] \\
%     V^\star (\vect s) &= Q^\star(\vect s, \vect \pi^\star(\vect s)) = \min_{\vect a} Q^\star(\vect s, \vect a)
% \end{align}

\subsection{Model Predictive Control}
An \ac{mpc} repeatedly solves a finite-horizon \ac{ocp} using a dynamics model to produce constrained control sequences, possibly subjected to constraints~\citep{rawlings2017model}. 
For a state $\vect s$, an \ac{mpc} scheme is typically cast as,
\begin{mini!}|s|
	{\vect u, \vect x}
	{\gamma^N T(\vect x_N) + \sum_{k=0}^{N-1} \gamma^k L(\vect x_k, \vect u_k)}
	{\label{eq:mpc1}}{}
	\addConstraint{\vect x_{k+1}}{= \vect f(\vect x_k, \vect u_k)}
	\addConstraint{\vect h(\vect x_k, \vect u_k)}{\leq 0}
	\addConstraint{\vect h^f(\vect x_N)}{\leq 0}
	\addConstraint{\vect x_0}{= \vect s}
\end{mini!}
where $L$ is the stage cost, $T$ is the terminal cost, $\vect f$ is the dynamics model approximating \eqref{eq:tf1}, and $\vect h, \vect h^f$ are the constraints of the controlled system, with the prediction horizon $N$ and the discount factor $\gamma$.
% In \eqref{eq:mpc1}, $\vect u$ and $\vect x$ collect the control input sequences and the state predictions, respectively, i.e. $\vect u = \{\vect u_0, \dots, \vect u_{N-1}\}$, $\vect x = \{\vect x_0, \dots, \vect x_{N}\}$. 
The solution of \eqref{eq:mpc1} yields an optimal control sequence $\vect u^\star$. % and state predictions $\vect x^\star$. 
An \ac{mpc} operates in a receding-horizon manner, applying only $\vect u_0^\star$ and resolving at the next state.
% An \ac{mpc} operates in a receding-horizon manner, i.e., only the first element of $\vect u^\star$, $\vect u_0^\star$, is applied to the system, and \eqref{eq:mpc1} is solved again for the resulting state. 
Then an \ac{mpc} policy is, % given as,
\begin{align}\label{eq:pi1}
	\vect \pi^{\textrm{MPC}}(\vect s) = \vect u_0^{\star}.
\end{align}
% where $\vect u_0^\star$ is obtained from solving \eqref{eq:mpc1}.
%The hope here is that $\vect \pi^{\textrm{MPC}}$ approximates closely $\vect \pi^\star$. 
The objective is for $\vect \pi^{\mathrm{MPC}}$ to approximate the optimal policy $\vect \pi^\star$ in \eqref{eq:optpi}. Next, we briefly discuss when $\vect \pi^{\textrm{MPC}}$ can deliver $\vect \pi^\star$ to yield a solution to an \ac{mdp}.

\subsection{\ac{mpc} as a solution to \ac{mdp}}
% Bottlenecks in MPC delivering MDP solution
In practice, $\vect \pi^{\textrm{MPC}}$ delivers satisfactory performance and enforces constraints provided that $\vect f$ adequately represents \eqref{eq:tf1} and $T$ and $N$ are appropriately chosen. 
% However, $\vect \pi^{\textrm{MPC}}$ often falls short of delivering optimal decisions due to the several simplifications made in formulating its optimization scheme from the original \ac{mdp} objective \eqref{eq:j1}, namely, the finite-horizon formulation, optimization over fixed input sequences $\vect u$ rather than a policy, and the inevitable inaccuracies in $\hat{\vect f}$. 
However, $\vect \pi^{\textrm{MPC}}$ often falls short of optimality due to simplifications in its formulation relative to \eqref{eq:j1}, namely, the finite-horizon formulation, optimization over fixed input sequences $\vect u$ rather than a policy, and approximating stochastic dynamics via a deterministic model $\vect f$.
While these simplifications are necessary for a tractable optimization problem, selecting $\vect f$ to achieve good closed-loop performance is particularly challenging. 
% Many real systems are complex to model accurately, and deterministic models are intrinsically inadequate in stochastic settings. Even when a stochastic dynamics model is used, \ac{mpc} does not fully capture the structure of the underlying \ac{mdp}, since it still optimizes over open-loop input sequences rather than a policy function.
% Many learning-based MPC approaches therefore seek to improve $\hat{f}$ from data using machine learning tools to capture uncertainty in the underlying dynamics, such as Gaussian Process MPC \cite{hewing2020gpmpc}, which models the transition dynamics through a predictive distribution over next-state transitions.
Models suited for decision-making generally differ from those optimized for prediction accuracy, except for specific classes of problems~\citep{anand2024data, anand2025all}. % as strictly dissipative problems \cite{anand2024data}. 
Taken together, these simplifications often cause an \ac{mpc} to fail to deliver an optimal policy for an \ac{mdp}. %thereby motivating the need to adapt the \ac{mpc} formulation to recover optimal decision-making performance.

Alternatively, an \ac{mpc}, even with an inaccurate dynamics model, can deliver an optimal policy for an \ac{mdp} if its stage and terminal costs are suitably modified~\citep{gros2019data}.
However, computing the required cost modifications is as complex as obtaining an accurate representation of \eqref{eq:tf1} and solving the underlying \ac{mdp}. 
Hence, \citep{gros2019data} further proposed parameterizing the \ac{mpc} scheme as a whole and learning its parameters using \ac{rl} techniques to optimize its closed-loop performance.
% Moreover, extensions of \ac{mpc} within an \ac{rl} framework have addressed settings without explicit state definitions by leveraging past input–output data for decision-making~\citep{sawant2023model}.
% These approach can then be viewed as i) adjusting \ac{mpc} parameters using \ac{rl}, or ii) using \ac{mpc} as a function approximator within an \ac{rl} pipeline to approximate optimal policy and value functions.
% Importantly, \citep{gros2019data} proposed updating the parameters to improve closed-loop performance rather than focusing on dynamics model accuracy, i.e., treating \ac{mpc} as a solution to an \ac{mdp}. 
This approach leads to the following parameterized \ac{mpc} for approximating a state–action value function,
\begin{mini!}|s|
	{\vect u, \vect x}
	{\gamma^N T_{\vect \theta}(\vect x_N) + \sum_{k=0}^{N-1} \gamma^k L_{\vect \theta}(\vect x_k, \vect u_k)}
	{\label{eq:mpc2}}
	{Q^{\mathrm{MPC}}_{\vect \theta}(\vect s, \vect a) =}
	\addConstraint{\vect x_{k+1}}{= \vect f_{\vect \theta}(\vect x_k, \vect u_k)}
	\addConstraint{\vect h_{\vect \theta}(\vect x_k, \vect u_k)}{\leq 0}
	\addConstraint{\vect h^f_{\vect \theta}(\vect x_N)}{\leq 0}
	\addConstraint{\vect x_0}{= \vect s}
	\addConstraint{\vect u_0}{= \vect a}
\end{mini!}
where $\vect \theta$ collects all the parameters introduced to the costs, dynamics, and constraints. 
The corresponding state value function and the \ac{mpc} policy are then given as, 
\begin{align}
    V^{\mathrm{MPC}}_{\vect \theta}(\vect s) =& \min_{\vect a} Q^{\mathrm{MPC}}_{\vect \theta}(\vect s, \vect a),\\
    \vect \pi^{\mathrm{MPC}}_{\vect \theta}(\vect s) =& \arg\min_{\vect a} Q^{\mathrm{MPC}}_{\vect \theta}(\vect s, \vect a) = \vect u^\star_0\,.
\end{align}
% where $\vect u_0$ is the solution of \eqref{eq:mpc2}. 
For brevity, we write $V_{\vect \theta}, Q_{\vect \theta}, \pi_{\vect\theta}$ to represent \ac{mpc}-based value functions and policy. %, $V^{\mathrm{MPC}}_{\vect \theta}, Q^{\mathrm{MPC}}_{\vect \theta}, \pi^{\mathrm{MPC}}_{\vect\theta}$.

Next, we discuss how future information is handled in \acp{mpc} and \acp{mdp}.
% Next, we briefly discuss the handling of future information in \acp{mpc} and \acp{mdp} when available.

\subsection{Future information in MPC and MDP}\label{sec:gen_mdp} 
In many real-world decision-making problems, partial information about the future, such as forecasts or reference trajectories, is available at decision time and can significantly improve decision quality.
%For example, if a production or sales forecast for the next period is known, it should be considered for the current decision, as it directly influences the optimal choice today. 
% Such future information must be explicitly accounted for when formulating control policies.
In the following section, we first review how \ac{mpc} typically handles available future information and then contrast it with the \ac{mdp} viewpoint. %, where future information must be incorporated in a way that preserves the Markov property. 

% Let $\vect z_{t, i}$ denote the available information for the $i$-th step in the future from the current time $t$, i.e., for the future time instance $i+t$. 
Let $\vect z_{t, i}$ denote the information available at time $t$ for the future time $t+i$.
% Then, \eqref{eq:mpc1} is typically trivially extended to account for $\vect z_{t,i}$ as,
Then, \eqref{eq:mpc1} is typically extended in a stage-wise manner as,
\begin{mini!}|s|
	{\vect u, \vect x}
	{\gamma^N T(\vect x_N, \vect z_{t, N}, \dots) + \sum_{k=0}^{N-1} \gamma^k L(\vect x_k, \vect u_k, \vect z_{t, k})}
	{\label{eq:mpc3}}{}
	\addConstraint{\vect x_{k+1}}{= \vect f(\vect x_k, \vect u_k, \vect z_{t, k})}
	\addConstraint{\vect h(\vect x_k, \vect u_k, \vect z_{t, k})}{\leq 0}
	\addConstraint{\vect h^f(\vect x_N, \vect z_{t, N})}{\leq 0}
	\addConstraint{\vect x_0}{= \vect s_t}
\end{mini!}
% where the terminal cost $T$ accounts for the remaining future information sequence, and modeling such a cost can get challenging.
% where the terminal cost $T$ accounts for the remaining future information beyond the prediction horizon, which can be difficult to model accurately.
where the terminal cost $T$ is expected to account for the effect of future information beyond the prediction horizon, which can be difficult to model accurately.

In an \ac{mdp}, the state should be \emph{Markovian} and contain all information necessary for optimal decision-making~\citep{puterman2014markov}. 
When future information is available, it can be interpreted as a latent representation of complex real-world data, condensed into information relevant for decision-making.
Defining an \ac{mdp} in such settings therefore necessitates an augmented state definition~\citep{powell2019rlso, de2021incorporating}.
Moreover, future information evolves according to its own stochastic dynamics and update mechanisms; for example, weather forecasts are periodically revised in response to new observations.

Consider the augmented state at time $t$ as, 
\begin{align} \label{eq:aug_state}
    \Tilde{\vect s}_t := [\,\vect s_t,\ \vect Z^H_{t,0}\,]^\top ,\, \Tilde{\vect s}_t \in \Tilde{\mathcal{S}},
\end{align}
where $\Tilde{\mathcal{S}}$ denotes the augmented state space and $\vect Z^H_{t,i} = [\vect z_{t,i}, \dots, \vect z_{t, i+H}]$ indicates future information sequence at $k$-th stage with $H$ as the length of relevant future information. 
In principle, $H$ may be arbitrarily large, but future information typically becomes less informative as its statistics regress toward their steady-state distribution over time. 
Thus, beyond a problem-dependent length, the choice of $H$ is a modeling decision, and only the relevant portion needs to be included.
% In principle, $H$ may be arbitrarily large, but future information typically becomes less informative as its statistics regress toward their steady state distribution over time; hence, only the relevant portion needs to be included. % arbitrary choice
% In principle, $H$ may be arbitrarily large, but as such information statistics typically regress to their base statistics over time and become uninformative for the current decision, only a relevant portion needs to be included. % Importantly, $H$ must be long enough to contain all future information on which the optimal decision depends.
The transition kernel $\Tilde{\mathcal{P}}$ then governs the evolution of $\Tilde{\vect s}_t$; it includes both the transition of the physical state $\vect{s}_t$ and the stochastic update of the future information $\vect Z^H_{t,0}$. 
With $\Tilde{\vect s}_t$ and $\Tilde{\mathcal{P}}$, an \ac{mdp} with future information can be defined for the underlying stochastic decision process. 
All associated value functions and policies for such an \ac{mdp} are then defined over $\Tilde{\vect s}_t$. 

% In the next section, we formulate an \ac{mpc} scheme over the augmented state and establish conditions under which it can yield optimal decisions for \acp{mdp} with future information.
% In the next section, we formulate an augmented-state \ac{mpc} scheme and establish conditions under which it can yield optimal decisions for \acp{mdp} with future information.
This motivates formulating \ac{mpc} directly over the augmented state, as developed in the next section.

\section{Optimal policies with MPC}\label{sec:mpc_optimal_policy}
In this section, we establish conditions under which a parameterized \ac{mpc} can yield optimal value functions and an optimal policy for an \ac{mdp} with future information.

% Including future information in the state is necessary for optimal decisions.
The \ac{mpc} in \eqref{eq:mpc3} incorporates future information in a stage-wise manner.
Such a formulation may be sufficient in special cases where the dynamics model is exact, and the terminal cost fully accounts for the remaining future information.
In practice, however, this requires accurately modeling the dynamics and capturing, through the terminal cost, the effect of future information beyond the prediction horizon; otherwise, an optimality gap may arise.

To address this, we build on the perspective in \citep{gros2019data}, treating \ac{mpc} as a parameterized model of the value functions and policy of an \ac{mdp} with future information.
Since the Markov property requires an augmented state, any \ac{mpc} intended to represent the optimal policy must also operate on $\Tilde{\vect s}$.
This augmented-state construction aligns the \ac{mpc} recursion with the Bellman recursion, which is essential for extending the optimality argument of \citep[Theorem~1]{gros2019data}.
This leads to the following parameterized \ac{mpc} scheme,
\begin{mini!}|s|
	{\vect u, \Tilde{\vect x}}
	{\gamma^N T_{\vect \theta}(\Tilde{\vect x}_N) + \sum_{k=0}^{N-1} \gamma^k L_{\vect \theta}(\Tilde{\vect x}_k, \vect u_k)}
	{\label{eq:mpc4}}
	{Q_{\vect \theta}(\Tilde{\vect s}_t, \vect a) =}
	\addConstraint{\Tilde{\vect x}_{k+1}}{= \Tilde{\vect f}(\Tilde{\vect x}_k, \vect u_k) \label{eq:mpc4:2}}
	\addConstraint{\vect h_{\vect \theta}(\Tilde{\vect x}_k, \vect u_k)}{\leq 0}
	\addConstraint{\vect h^f_{\vect \theta}(\Tilde{\vect x}_N)}{\leq 0}
	\addConstraint{\Tilde{\vect x}_0}{= \Tilde{\vect s}_t}
	\addConstraint{\vect u_0}{= \vect a}
\end{mini!}
where $\Tilde{\vect f}$ is a deterministic model of the transition kernel $\Tilde{\mathcal{P}}$ and propagates the predicted augmented state $\Tilde{\vect x}_{k}$.
Then the state value function and the \ac{mpc} policy are,
\begin{align}
	V_{\vect \theta}(\Tilde{\vect s}_t) =& \min_{\vect a} Q_{\vect \theta}(\Tilde{\vect s}_t, \vect a), \label{eq:vmpc}\\
\vect \pi_{\vect\theta}(\Tilde{\vect s}_t) =& \arg\min_{\vect a} Q_{\vect \theta}(\Tilde{\vect s}_t, \vect a) = \vect u_0^\star. \label{eq:pi2}
\end{align}

We now show that the parameterized \ac{mpc} scheme in \eqref{eq:mpc4}--\eqref{eq:pi2} can represent the optimal value functions and policy of an \ac{mdp} with future information.

\begin{theorem} \label{th:rlmpc}
Consider the parameterized stage cost, terminal cost, and constraints in \eqref{eq:mpc4}--\eqref{eq:pi2} as function approximators with parameters $\vect\theta$. Let $\Tilde{\vect x}^{\star}=\{\Tilde{\vect x}^{\star}_k\}_{k=0}^N$ denote the predicted augmented state trajectory generated by the dynamics model $\Tilde{\vect f}$ under the optimal policy $\vect\pi^\star$.  

Then, there exist parameters $\vect\theta^\star$ such that:
\begin{enumerate}[label=(\arabic*), ref=(\arabic*)]
    \item \label{th:rlmpc:V}
    $V_{\vect\theta^\star}(\Tilde{\vect s}_t) = V^\star(\Tilde{\vect s}_t),\quad \forall\, \Tilde{\vect s}_t \in \Tilde{\mathcal X}$,
    \item \label{th:rlmpc:pi}
    $\vect\pi_{\vect\theta^\star}(\Tilde{\vect s}_t) = \vect\pi^\star(\Tilde{\vect s}_t),\quad \forall\, \Tilde{\vect s}_t \in \Tilde{\mathcal X}$,
    \item \label{th:rlmpc:Q}
    $Q_{\vect\theta^\star}(\Tilde{\vect s}_t,\vect a) = Q^\star(\Tilde{\vect s}_t,\vect a)$ for all $\Tilde{\vect s}_t \in \Tilde{\mathcal X}$ and all $\vect a \in \mathcal A$ such that  
    $\left| V^\star \left(\Tilde{\vect f}(\Tilde{\vect s}_t,\vect a)\right) \right| < \infty$.
\end{enumerate}
where
\begin{equation}
    \Tilde{\mathcal X} := \left\{
\Tilde{\vect s}_t \in \Tilde{\mathcal S} \;\middle|\;
\left| V^\star(\Tilde{\vect x}^\star_k) \right| < \infty \,\text{for all}\, k \le N\right\}.
\end{equation}
\end{theorem}

\begin{pf}
    Consider $\vect\theta$ satisfying the following, 
\begin{subequations}\label{eq:theta1}
\begin{align}
L_{\vect\theta}(\Tilde{\vect s}_t,\vect a) &= \begin{cases}
Q^\star(\Tilde{\vect s}_t,\vect a) - \gamma V^+(\Tilde{\vect s}_t,\vect a), 
&\text{if } |V^+(\Tilde{\vect s}_t,\vect a)| < \infty, \\[2mm]
\infty, &\text{otherwise},
\end{cases} \\
T_{\vect\theta}(\Tilde{\vect s}_t) &= V^\star(\Tilde{\vect s}_t),
\end{align}
\end{subequations}
where $V^+(\Tilde{\vect s}_t,\vect a) = V^\star (\Tilde{\vect f}(\Tilde{\vect s}_t,\vect a))$.
Since $\Tilde{\vect f}$ is a deterministic model of $\Tilde{\mathcal{P}}$, the Bellman recursion unfolds as a telescoping sum along the predicted trajectory.  
The result then follows from \cite[Theorem~1]{gros2019data}. \hfill \qed
\end{pf}

Theorem~\ref{th:rlmpc} characterizes the existence of an \ac{mpc} parameterization that yields the optimal value functions and policy. 
The augmented state formulation is essential for this result: it allows the stage and terminal costs in \eqref{eq:theta1} to depend on the same information as the optimal value functions of the underlying \ac{mdp}. 
Consequently, suitable cost modifications can reproduce the Bellman recursion and compensate for inaccuracies in the prediction model, as in \citep{gros2019data}, for \acp{mdp} with future information.
% As an immediate consequence, we obtain the following.
% \begin{corollary}%[MPC as a function approximator]
% \label{cor:mpc_fa}
% Under the assumptions of Theorem~\ref{th:rlmpc}, a sufficiently parameterized \ac{mpc} scheme of the forms \eqref{eq:mpc4}-\eqref{eq:pi2} constitutes a valid function approximator to represent the optimal value functions and the optimal policy of an \ac{mdp} with future information. 
% In particular, when richly parameterized, there exist $\vect \theta^\star$ such that,
% \begin{align}
%     &V_{\vect \theta^\star}(\Tilde{\vect s}) = V^\star(\Tilde{\vect s}),\\ 
%     &\vect \pi_{\vect \theta^\star}(\Tilde{\vect s}) = \vect \pi^\star(\Tilde{\vect s}), \\
%     &Q_{\vect \theta^\star}(\Tilde{\vect s}, \vect a) = Q^\star(\Tilde{\vect s}, \vect a),
% \end{align}
% over the set $\Tilde{\mathcal{X}}$.
% \end{corollary}
% Comments on the theorem
% Theorem~\ref{th:rlmpc} establishes that parameterized \ac{mpc}s of the forms \eqref{eq:mpc4}-\eqref{eq:pi2} can represent the optimal value functions and the optimal policy of an \ac{mdp} with future information. 
This result implies several requirements on how the \ac{mpc} formulation must be constructed:
\begin{itemize}[leftmargin=*]
    \item An \ac{mpc} scheme must operate on the augmented state:
    Since future information influences optimal decisions, the \ac{mpc} components must all be defined over $\tilde{\vect s}$, motivating the \ac{mpc} formulation in \eqref{eq:mpc4}.
    In contrast to \eqref{eq:mpc3}, \eqref{eq:mpc4} thus aligns with the Bellman structure of an \ac{mdp} with future information and can, in principle, represent its optimal value functions and policies. 
    % This may require a blending of the future information in time at every stage of the \ac{mpc} horizon, unlike the simpler formulation in \eqref{eq:mpc3}. %, which uses only the $k$-step-ahead forecast $z_k$.

    \item Optimality is defined in the augmented state space:
    The Bellman operator is well-posed only when defined over $\Tilde{\vect s}$. Any \ac{mpc} or \ac{rl} architecture that ignores the forecast component of the state violates the Markov property and hence cannot recover $V^\star$, $Q^\star$, or $\vect \pi^\star$ in theory.

    \item Value function approximators must use the augmented state:
    Any critic must approximate $V^{\vect \pi}(\Tilde{\vect s})$ or $Q^{\vect \pi}(\Tilde{\vect s}, \vect a)$, rather than functions of the physical state alone. %not their counterparts defined on the physical state alone.

    \item Future information must be propagated consistently:
    To match the Bellman recursion, the \ac{mpc} prediction model must evolve the full augmented state, updating $\vect Z^H_{t,k}$ to $\vect Z^H_{t,k+1}$ as in \eqref{eq:mpc4:2}. 
    A simple strategy is to shift a precomputed sequence of future information.
    % To match the Bellman recursion, the \ac{mpc} prediction model must evolve the full augmented state, including the future information sequence $\vect Z^H_{t, k}$, as in \eqref{eq:mpc4:2}. 
    % In particular, the transition from stage $k$ to $k+1$ must update $\vect Z^H_{t,k}$ to the next sequence $\vect Z^H_{t,k+1}$.
    % A simple strategy is to shift a precomputed future information sequence. %, so that $\vect Z^H_{t,k+1} = [\vect z_{t,k+1},\dots,\vect z_{t,k+H+1}]$.
    % Practical implementations are discussed in Section~\ref{sec:mpc_fn_approx}.
    
    % \item Length of future information vs.\ the \ac{mpc} horizon:
    % The future-information window $H$ and the \ac{mpc} horizon $N$ are independent.
    % For $H>0$, the formulation in \eqref{eq:mpc4} can exploit information beyond the current \ac{mpc} stage, and possibly beyond the \ac{mpc} horizon itself.
    % If propagation is implemented by shifting a precomputed future-information sequence, the sequence must be long enough to supply all stage-wise windows; using a window of length $H$ over a horizon $N$ requires information up to approximately $N+H$ steps ahead. % maybe not necessary
\end{itemize}

In summary, Theorem~\ref{th:rlmpc} shows that a sufficiently parameterized \ac{mpc} can represent the optimal value functions and policy of an \ac{mdp} with future information, provided it operates on the augmented state and propagates forecast information through its horizon.  
Although finding the optimal parameters is as hard as solving the \ac{mdp} itself, this result motivates treating \acp{mpc} in \eqref{eq:mpc4}--\eqref{eq:pi2} as function approximators. 
Section~\ref{sec:mpc_fn_approx} presents a practical implementation and racing example.

\section{MPC as a function approximator}\label{sec:mpc_fn_approx}
% This section presents a practical \ac{mpc} formulation for \acp{mdp} with future information and its use as a function approximator within an \ac{rl} pipeline.
This section presents a decomposed augmented-state \ac{mpc} formulation for \acp{mdp} with future information and its use as a function approximator within an \ac{rl} pipeline.

% In \eqref{eq:mpc4}, the prediction model $\Tilde{\vect f}$ acts on the augmented state. 
In practice, the evolution of future information is often independent of the \ac{mpc} optimization variables. 
This motivates decomposing $\Tilde{\vect f}$ in \eqref{eq:mpc4} into the physical state dynamics and the future information update, represented by $\vect f_{\vect\theta}$ and $\vect g_{\vect\theta}$, respectively. 
We parameterize both components to provide \ac{rl} with additional flexibility to tune $\vect\theta$ for achieving \ref{th:rlmpc:V}--\ref{th:rlmpc:Q} of Theorem~\ref{th:rlmpc}. 
We then formulate a parameterized \ac{mpc} scheme as,
\begin{mini!}|s|
	{\vect u, \vect x}
	{\gamma^N T_{\vect \theta}(\vect x_N, \vect Z_{N}) + \sum_{k=0}^{N-1} \gamma^k L_{\vect \theta}(\vect x_k, \vect u_k, \vect Z_{k})}
	{\label{eq:mpc5}}
	{Q_{\vect \theta}(\Tilde{\vect s}_t, \vect a) =}
	\addConstraint{\vect x_{k+1}}{= \vect f_{\vect \theta}(\vect x_k, \vect u_k, \vect Z_{k})}
	\addConstraint{\vect Z_{k+1}}{= \vect g_{\vect \theta}(\vect Z_{k})}
	\addConstraint{\vect h_{\vect \theta}(\vect x_k, \vect u_k, \vect Z_{k})}{\leq 0}
	\addConstraint{\vect h^f_{\vect \theta}(\vect x_N, \vect Z_{N})}{\leq 0}
	\addConstraint{[\vect x_0, \vect Z_{0}]}{= \Tilde{\vect s}_t}
	\addConstraint{\vect u_0}{= \vect a}
\end{mini!}
where $\vect Z_k := \vect Z^H_{t,k}$ within the \ac{mpc} horizon, for compactness. 
Then the state value function and policy are,
\begin{align}
	V_{\vect \theta}(\Tilde{\vect s}_t) =& \min_{\vect a} Q_{\vect \theta}(\Tilde{\vect s}_t, \vect a), \label{eq:vmpc2}\\
\vect \pi_{\vect\theta}(\Tilde{\vect s}_t) =& \arg\min_{\vect a} Q_{\vect \theta}(\Tilde{\vect s}_t, \vect a) = \vect u_0^\star. \label{eq:pi3}
\end{align}

The formulation in \eqref{eq:mpc5} is a structured instance of \eqref{eq:mpc4}. 
Theorem~\ref{th:rlmpc} does not require $\Tilde{\vect f}$ to match the true transition kernel $\Tilde{\mathcal P}$ exactly; hence, it remains applicable even when $\vect f_{\vect\theta}$ and $\vect g_{\vect\theta}$ are approximate or learned models. 
In particular, with suitable cost modifications of the form \eqref{eq:theta1}, an \ac{mpc} of the form \eqref{eq:mpc5} can still represent the optimal value functions and policy of the \ac{mdp} with future information.

This decomposition is useful in practice as the future information update can often be handled outside the optimization problem. 
For example, $\vect g_{\vect\theta}$ may simply shift a precomputed future information sequence, so that each stage receives the appropriate sequence $\vect Z^H_{t,k}$. 
Alternatively, $\vect g_{\vect\theta}$ may be parameterized and tuned together with the rest of the \ac{mpc} scheme using \ac{rl}. 
% Thus, \eqref{eq:mpc5} retains the augmented state structure required by Theorem~\ref{th:rlmpc}, while avoiding explicit optimization over the future information dynamics.
% Thus, \eqref{eq:mpc5} retains the augmented-state structure required by Theorem~\ref{th:rlmpc}, while allowing practical implementations that avoid explicitly optimizing over the future-information dynamics.

Next, we describe how the parameterized \acp{mpc} can be embedded within an \ac{rl} pipeline.

\subsection{RL and MPC sensitivities}
Learning a parameterized \ac{mpc} with \ac{rl} requires gradients of the value functions and policy with respect to $\vect\theta$. 
For the \ac{mpc} schemes in \eqref{eq:mpc5}-\eqref{eq:pi3}, these gradients are obtained through standard \ac{mpc} sensitivity analysis by differentiating the Lagrangian and KKT conditions; see \citep{gros2019data} for details.
Here, the sensitivities are computed for \acp{mpc} defined over the augmented state $\Tilde{\vect s}$.
Consequently, the gradients account for the dependence of the value functions and policy on both the physical state and the available future information.

These sensitivities allow the parameterized \ac{mpc} to be embedded in standard \ac{rl} updates. 
For example, in a policy-gradient method such as \ac{dpg}~\citep{silver2014deterministic}, the \ac{mpc} in \eqref{eq:pi2} directly serves as a policy and is updated using,
\begin{equation}
    \vect \theta \leftarrow \vect \theta + \alpha \nabla_{\vect \theta} \vect \pi_{\vect \theta}(\Tilde{\vect s}) \nabla_{\vect a} Q^{\vect \pi_{\vect \theta}}(\Tilde{\vect s}, \vect a)|_{\vect a = \vect \pi_{\vect \theta}(\Tilde{\vect s})},
\end{equation}
where $Q^{\vect \pi_{\vect \theta}}$ is estimated separately.
In stochastic policy-gradient methods~\citep{sutton1999policy}, the same construction can be used by inducing a stochastic policy, for instance by randomly perturbing the \ac{mpc} cost objective as in \citep{gros2021reinforcement} or by simply sampling actions from a Gaussian distribution centered at the \ac{mpc} output.

The \ac{mpc} schemes in \eqref{eq:mpc5}--\eqref{eq:pi3} incorporate future information into learning and decision-making through their structured optimization problems.
Since this information enters through the \ac{mpc} cost, dynamics, and constraints, the resulting policy can exploit future information meaningfully even before learning.
This provides a useful inductive bias compared with generic policy approximators that receive the augmented state directly and must infer the relevant structure from data.

In actor--critic methods, a separate value function approximator is still required.
Neural networks remain a natural choice for this role due to their expressiveness, but they must approximate $V^{\vect \pi}(\Tilde{\vect s})$ or $Q^{\vect \pi}(\Tilde{\vect s}, \vect a)$ over the augmented state.
Thus, while the \ac{mpc} structure provides a strong inductive bias for the policy, value-function learning may still suffer from the increased dimensionality of $\Tilde{\vect s}$.

\subsection{Illustrative example}
We evaluate the proposed \ac{mpc} in \eqref{eq:pi3} on a point-mass racing task. % to demonstrate its capabilities as a function approximator for an \ac{mdp} with future information. 
In this task, a sequence of future reference positions (the track centerline) is available at decision time and constitutes the future information component of the augmented state. 
We deploy the parameterized \ac{mpc} within a \ac{ppo}~\citep{schulman2017proximal} training loop and show that the learned policy exploits this future information to achieve high-performing closed-loop behavior.

\subsubsection{Point-Mass Racing Environment:}
The controlled system is a planar point-mass with physical state, $\vect{s} = [p_x, p_y, v_x, v_y]^\top$ where $\vect p = (p_x,p_y)$ are positions and $\vect v = (v_x,v_y)$ are velocities.  
The control input is the applied force $\vect{a} = [F_x, F_y]^\top$. 
The dynamics follow a discretized double-integrator with Gaussian noise,
\begin{equation*}
    \vect{s}_{t+1} = A\, \vect{s}_t + B\, \vect{a}_t + \mathcal{N}(0, 10^{-3}I),
\end{equation*}
with
\begin{equation*}
A=\begin{bmatrix}
1 & 0 & 0.1 & 0\\
0 & 1 & 0 & 0.1\\
0 & 0 & 0.9 & 0\\
0 & 0 & 0 & 0.9
\end{bmatrix},\,
B=\begin{bmatrix}
0 & 0\\
0 & 0\\
0.1 & 0\\
0 & 0.1
\end{bmatrix}.
\end{equation*}

The goal is to traverse a racing track as quickly as possible (as shown in fig. \ref{fig:traj}). The environment provides a sequence of future reference positions, $\vect{Z}^M_{t, 0} = [\, \vect z_{t, 0},\, \dots,\, \vect z_{t, M}]$, where $\vect z_{t, 0} \in \mathbb{R}^2$ denotes the track centerline closest to $\vect s_t$ with $M=60$.
Here, we keep the future information length at $H=50$ (i.e., $\tilde{\vect s}_t = [\vect s_t, \vect Z^{50}_{t, 0}]$) with the rest used to shift $\vect Z^H_{t, k}$ over the optimization horizon. 
% The first sequence constitutes the future information component of the augmented state.

The task is encoded through the stage cost,
% \begin{align}\label{eq:reward}
%     L(\tilde{\vect{s}}_t,\vect{a}_t) =& -\mathrm{prog}(\vect{s}_t) + 0.01 \| \vect p_t - \vect z_{0,t} \|^2  \nonumber \\
%     & + 0.01 \| \vect{v}_t \|^2 + 0.01 \| \vect{a}_t \|^2,
% \end{align}
\begin{align}\label{eq:reward}
    L(\tilde{\vect{s}}_t,\vect{a}_t) =& -\mathrm{prog}(\vect{s}_t),
\end{align}
where $\mathrm{prog}(\cdot)$ measures forward progress along the track. % and $\vect z_t$ is the closest reference point.  
Exceeding the track half-width of $0.4$ terminates the episode. % and yields a penalty of $10 \| \vect p_t - \vect z_{0,t} \|^2$.
% Note that, in \eqref{eq:reward}, the stage cost is given as a reward to maximize, in line with \ac{rl} convention. 
Additionally, the discount factor $\gamma$ is set to $0.99$.

\subsubsection{MPC formulation:}
We adopt the parameterized \ac{mpc} policy in \eqref{eq:pi3} (with $N = 10$).  
% At each stage $k$, the future information sequence $\vect Z^H_{t,k}$ is constructed by shifting. %; thus, the evolution of future information is handled externally.
For demonstration, we parameterize only the \ac{mpc} stage cost as,
\begin{align*}
    L_{\vect \theta}(\vect{x}_k,\vect{u}_k,\vect{Z}^H_{t, k})
    &= W_{\vect p}\| \vect p_k - \vect Z^H_{t, k} \, \vect w_k \|^2 \\
     &\quad + W_{\vect v} \| \vect{v}_k \|^2  + W_{\vect u} \| \vect{u}_k \|^2 ,
\end{align*}
% where $\vect w_k \in \mathbb{R}^{H+1}$ is a weight vector forming a weighted combination of columns of $\vect Z^H_{t, k}$, i.e., $\vect Z^H_{t, k} \vect w_k = \sum_{i=0}^{H} w_{k, i} \, \vect z_{t, k+i}$.
% where $\vect w_k \in \mathbb{R}^{H+1}$ forms a latent reference $\vect z^{\mathrm{ref}}_k = \vect Z^H_{t,k}\vect w_k$ from the future information sequence $\vect Z^H_{t,k}$.
where $\vect w_k \in \mathbb{R}^{H+1}$ weights the columns of $\vect Z^H_{t,k}$ to form the latent reference $\vect Z^H_{t,k}\vect w_k$, i.e., $\vect Z^H_{t, k} \vect w_k = \sum_{i=0}^{H} w_{k, i} \, \vect z_{t, k+i}$.
Note that the term $\vect Z^H_{t, k} \vect w_{k}$ acts as a latent reference and can be tuned by updating $\vect w_k$, enabling the \ac{mpc} to exploit future information more flexibly.
Then, $\vect \theta = \{ W_{\vect p}, W_{\vect z}, W_{\vect v}, W_{\vect u} \}$ with $W_{\vect z}=[\vect w_0, \ldots, \vect w_N]^\top$ forms the parameter vector to be learned with \ac{ppo}. 
% This allows the policy to learn track-shaping behaviors naturally from data.

We compare three controllers:
\begin{enumerate}
    \item Nominal \ac{mpc}: A baseline \ac{mpc} controller using the naive formulation in \eqref{eq:mpc3} with fixed cost weights.
    \item \ac{ppo}–\ac{mpc}: A parameterized \ac{mpc} with only the cost weights, $\{ W_{\vect p}, W_{\vect v}, W_{\vect u} \}$, learned using \ac{ppo}, and naive future information usage, similar to \eqref{eq:mpc3}.
    \item \ac{ppo}–\ac{mpc}i: The proposed parameterized \ac{mpc} in \eqref{eq:pi3} with $\vect \theta = \{ W_{\vect p}, W_{\vect z}, W_{\vect v}, W_{\vect u} \}$ learned using \ac{ppo}. % to learn its own latent track representation.
\end{enumerate}

A stochastic policy is formed by sampling actions from a Gaussian distribution centered around the \ac{mpc} output,
\begin{equation*}
    \vect \pi^\prime_{\vect \theta}(\Tilde{\vect s}_t)
    = \mathcal{N}\!\left(
        \vect{\pi}_{\vect \theta}(\Tilde{\vect s}_t),
        \vect \sigma^2 I
    \right),
\end{equation*}
with $\vect \sigma$ as the standard deviation.

\subsubsection{PPO Training:}
The actor is an \ac{mpc}-based stochastic policy $\vect \pi^\prime_{\vect \theta}$, while the critic is a two-layer neural network with $128$ units per layer, trained to approximate $V^{\vect \pi^\prime_{\vect \theta}}(\tilde{\vect{s}})$.
The \ac{ppo} objective employs clipped likelihood ratios and generalized advantage estimation.  
Training is conducted for $1\times 10^4$ time steps with minibatch SGD.

\subsubsection{Discussion:}

\begin{figure}[t]
\begin{subfigure}{4.35cm}
    \centering
    \includegraphics[width=\textwidth]{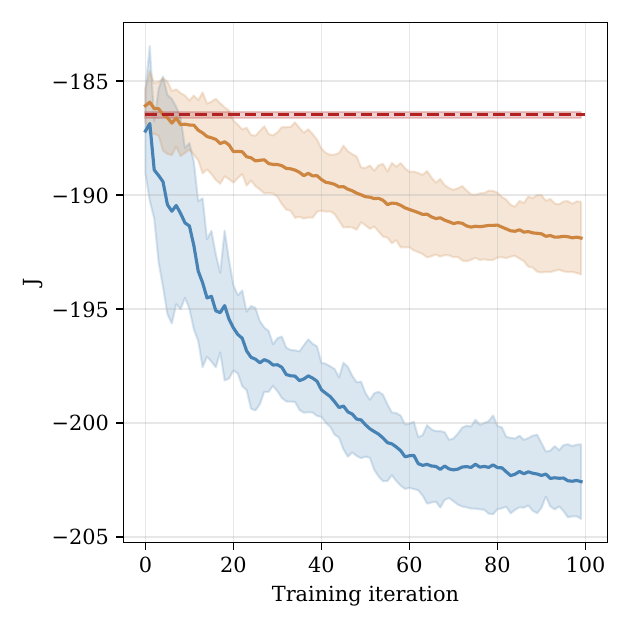}    % The printed column width is 8.4 cm.
    \caption{Cumulative Cost}
    \label{fig:convergence}
\end{subfigure}
\begin{subfigure}{4.35cm}
    \centering
    \includegraphics[width=\textwidth]{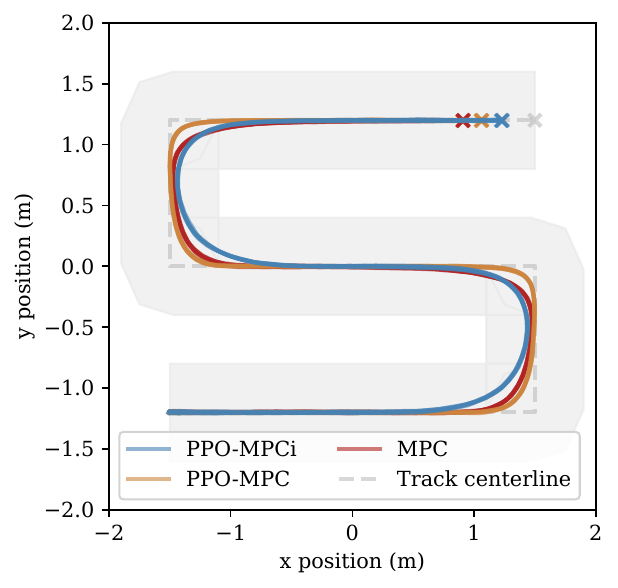}    % The printed column width is 8.4 cm.
    \caption{State Trajectory}
    \label{fig:traj}
\end{subfigure}
\caption{Closed-loop performance and state trajectory for (i) \ac{mpc}, (ii) \ac{ppo}-\ac{mpc}, and (iii) \ac{ppo}-\ac{mpc}i} 
%\caption{Cumulative return and closed-loop behavior for (i) nominal \ac{mpc}, (ii) \ac{ppo}-\ac{mpc}: \ac{mpc} learned with \ac{ppo} without blending $\vect d_{k,H}$, and (iii) \ac{ppo}-\ac{mpc}i: \ac{mpc} learned with \ac{ppo} with blending $\vect d_{k,H}$} 
\label{fig:overall}
\end{figure}

% Fig.~1a shows the cumulative cost during training, averaged over five seeds.
% PPO-MPCi achieves the lowest cost, followed by PPO-MPC, while the nominal MPC performs worst.
Fig.~\ref{fig:convergence} shows the cumulative cost during training, averaged over five seeds.  
\ac{ppo}-\ac{mpc}i achieves the lowest cost, followed by \ac{ppo}-\ac{mpc}, while the nominal \ac{mpc} performs worst.
The trajectories in Fig.~\ref{fig:traj} show that nominal \ac{mpc} follows the track conservatively, whereas \ac{ppo}-\ac{mpc} learns more aggressive tracking and reaches farther along the track, as indicated by the $\times$ markers.
In contrast, \ac{ppo}-\ac{mpc}i learns a latent reference through $\vect Z^H_{t,k}\vect w_k$, allowing it to exploit the future reference information sequence more flexibly and cut corners to increase progress.

The latent representation used here is only a weighted average over the future reference information sequence, yet already yields a clear performance gain.
In more complex domains, richer latent representations, such as variational autoencoders~\citep{Kingma_2019}, could compress long-horizon future information into low-dimensional features suitable for \ac{mpc}.
Such settings may also require alternative critic architectures, since the augmented state can become high-dimensional.

\section{Conclusion}\label{sec:conclusion}
This work establishes conditions under which a parameterized \ac{mpc} scheme can represent the optimal value functions and optimal policy of an \ac{mdp} with future information. 
In particular, we show that future information must be incorporated through the augmented state and propagated through the optimization horizon for the \ac{mpc} to align with the Bellman structure of the underlying \ac{mdp}. 
This formulation naturally enables \ac{mpc} to serve as a structured function approximator within an \ac{rl} pipeline.
We further demonstrate, through a racing example, that \ac{rl} can learn the \ac{mpc} parameters in practice and that the learned \ac{mpc} exploits future information more effectively than naive use of future information. 
Overall, the results highlight \ac{mpc} as an effective structured function approximator for decision problems with future information.

% This work establishes conditions under which a parameterized \ac{mpc} scheme can represent the optimal value functions and optimal policy of an \ac{mdp} with future information. 
% In particular, we show that future information must be incorporated through the augmented state and propagated through the optimization horizon for the \ac{mpc} to align with the Bellman structure of the underlying \ac{mdp}. 
% With sufficient parameterization, the \ac{mpc} scheme can therefore serve as a structured function approximator within an \ac{rl} pipeline.
% We further demonstrate, through a racing example, that \ac{rl} can learn the \ac{mpc} parameters in practice and that the learned \ac{mpc} exploits future information more effectively than naive use of future information. 
% Overall, the results highlight \ac{mpc} as an effective structured function approximator for decision problems with future information.

% \begin{ack}
% Place acknowledgments here.
% \end{ack}

\section*{DECLARATION OF GENERATIVE AI AND AI-ASSISTED TECHNOLOGIES IN THE WRITING PROCESS}
During the preparation of this work, the authors used ChatGPT to refine the text. After using this tool/service, the authors reviewed and edited the content as needed and take full responsibility for the content of the publication.

\bibliography{references}             % bib file to produce the bibliography

@book{rawlings2017model,
  title={Model predictive control: theory, computation, and design},
  author={Rawlings, James Blake and Mayne, David Q and Diehl, Moritz},
  volume={2},
  year={2017},
  publisher={Nob Hill Publishing Madison, WI}
}

@book{sutton2018reinforcement,
  title={Reinforcement learning: An introduction},
  author={Sutton, Richard S and Barto, Andrew G},
  year={2018},
  publisher={MIT press}
}

@article{gros2019data,
  title={Data-driven economic NMPC using reinforcement learning},
  author={Gros, S{\'e}bastien and Zanon, Mario},
  journal={IEEE Transactions on Automatic Control},
  volume={65},
  number={2},
  pages={636--648},
  year={2019},
  publisher={IEEE}
}

@ARTICLE{hewing2020gpmpc,
  author={Hewing, Lukas and Kabzan, Juraj and Zeilinger, Melanie N.},
  journal={IEEE Transactions on Control Systems Technology}, 
  title={Cautious Model Predictive Control Using Gaussian Process Regression}, 
  year={2020},
  volume={28},
  number={6},
  pages={2736-2743}
}

@InProceedings{wabersich2020bmpc,
  title = 	 {Bayesian model predictive control: Efficient model exploration and regret bounds using posterior sampling},
  author =       {Wabersich, Kim Peter and Zeilinger, Melanie},
  booktitle = 	 {Proceedings of the 2nd Conference on Learning for Dynamics and Control},
  pages = 	 {455--464},
  year = 	 {2020},
  volume = 	 {120},
  series = 	 {Proceedings of Machine Learning Research},
  month = 	 {10--11 Jun},
  publisher =    {PMLR},
}

@article{sutton1999policy,
  title={Policy gradient methods for reinforcement learning with function approximation},
  author={Sutton, Richard S and McAllester, David and Singh, Satinder and Mansour, Yishay},
  journal={Advances in neural information processing systems},
  volume={12},
  year={1999}
}

@inproceedings{silver2014deterministic,
  title={Deterministic policy gradient algorithms},
  author={Silver, David and Lever, Guy and Heess, Nicolas and Degris, Thomas and Wierstra, Daan and Riedmiller, Martin},
  booktitle={International conference on machine learning},
  pages={387--395},
  year={2014},
  organization={PMLR}
}

@article{hewing2020learning,
  title={Learning-based model predictive control: Toward safe learning in control},
  author={Hewing, Lukas and Wabersich, Kim P and Menner, Marcel and Zeilinger, Melanie N},
  journal={Annual Review of Control, Robotics, and Autonomous Systems},
  volume={3},
  number={1},
  pages={269--296},
  year={2020},
  publisher={Annual Reviews}
}

@INPROCEEDINGS{mesbah2022fusion,
	author={Mesbah, Ali and Wabersich, Kim P. and Schoellig, Angela P. and Zeilinger, Melanie N. and Lucia, Sergio and Badgwell, Thomas A. and Paulson, Joel A.},
	booktitle={2022 American Control Conference (ACC)}, 
	title={Fusion of Machine Learning and MPC under Uncertainty: What Advances Are on the Horizon?}, 
	year={2022},
	volume={},
	number={},
	pages={342-357},
}

@article{anand2024data,
  title={Data-driven predictive control and MPC: Do we achieve optimality?},
  author={Anand, Akhil S and Sawant, Shambhuraj and Reinhardt, Dirk and Gros, Sebastien},
  journal={IFAC-PapersOnLine},
  volume={58},
  number={15},
  pages={73--78},
  year={2024},
  publisher={Elsevier}
}

@incollection{reinhardt2025economic,
  title={Economic model predictive control as a solution to markov decision processes},
  author={Reinhardt, Dirk and Anand, Akhil S and Sawant, Shambhuraj and Gros, S{\'e}bastien},
  booktitle={Model Predictive Control: Engineering Methods for Economists},
  pages={157--189},
  year={2025},
  publisher={Springer}
}

@article{anand2025all,
  title={All AI Models are Wrong, but Some are Optimal},
  author={Anand, Akhil S and Sawant, Shambhuraj and Reinhardt, Dirk and Gros, Sebastien},
  journal={arXiv preprint arXiv:2501.06086},
  year={2025}
}

@article{schulman2017proximal,
  title={Proximal policy optimization algorithms},
  author={Schulman, John and Wolski, Filip and Dhariwal, Prafulla and Radford, Alec and Klimov, Oleg},
  journal={arXiv preprint arXiv:1707.06347},
  year={2017}
}

@misc{reiter2025synthesismodelpredictivecontrol,
      title={Synthesis of Model Predictive Control and Reinforcement Learning: Survey and Classification}, 
      author={Rudolf Reiter and Jasper Hoffmann and Dirk Reinhardt and Florian Messerer and Katrin Baumgärtner and Shamburaj Sawant and Joschka Boedecker and Moritz Diehl and Sebastien Gros},
      year={2025},
      eprint={2502.02133},
      archivePrefix={arXiv},
      primaryClass={eess.SY},
      url={https://arxiv.org/abs/2502.02133}, 
}

@book{puterman2014markov,
  title={Markov decision processes: discrete stochastic dynamic programming},
  author={Puterman, Martin L},
  year={2014},
  publisher={John Wiley \& Sons}
}

@article{Kingma_2019,
   title={An Introduction to Variational Autoencoders},
   volume={12},
   ISSN={1935-8245},
   number={4},
   journal={Foundations and Trends in Machine Learning},
   publisher={Emerald},
   author={Kingma, Diederik P. and Welling, Max},
   year={2019},
   pages={307–392} }

@inproceedings{romero2024actor,
  title={Actor-critic model predictive control},
  author={Romero, Angel and Song, Yunlong and Scaramuzza, Davide},
  booktitle={2024 IEEE International Conference on Robotics and Automation (ICRA)},
  pages={14777--14784},
  year={2024},
  organization={IEEE}
}

@article{de2021incorporating,
  title={On incorporating forecasts into linear state space model Markov decision processes},
  author={de Chalendar, Jacques A and Glynn, Peter W},
  journal={Philosophical Transactions of the Royal Society A},
  volume={379},
  number={2202},
  pages={20190430},
  year={2021},
  publisher={The Royal Society Publishing}
}

@book{powell2019rlso,
  title={Reinforcement Learning and Stochastic Optimization: A Unified Framework for Sequential Decisions},
  author={Powell, Warren B.},
  year={2019},
  publisher={John Wiley \& Sons}
}

@inproceedings{gros2021reinforcement,
  title={Reinforcement learning based on MPC and the stochastic policy gradient method},
  author={Gros, S{\'e}bastien and Zanon, Mario},
  booktitle={2021 American Control Conference (ACC)},
  pages={1947--1952},
  year={2021},
  organization={IEEE}
}
                                                     % with bibtex (preferred)
                                                   
%\begin{thebibliography}{xx}  % you can also add the bibliography by hand

%\bibitem[Able(1956)]{Abl:56}
%B.C. Able.
%\newblock Nucleic acid content of microscope.
%\newblock \emph{Nature}, 135:\penalty0 7--9, 1956.

%\bibitem[Able et~al.(1954)Able, Tagg, and Rush]{AbTaRu:54}
%B.C. Able, R.A. Tagg, and M.~Rush.
%\newblock Enzyme-catalyzed cellular transanimations.
%\newblock In A.F. Round, editor, \emph{Advances in Enzymology}, volume~2, pages
%  125--247. Academic Press, New York, 3rd edition, 1954.

%\bibitem[Keohane(1958)]{Keo:58}
%\newblock \emph{Power and Interdependence: World Politics in Transitions}.
%\newblock Little, Brown \& Co., Boston, 1958.

%\bibitem[Powers(1985)]{Pow:85}
%\newblock Is there a way out?
%\newblock \emph{Harpers}, pages 35--47, June 1985.

%\bibitem[Soukhanov(1992)]{Heritage:92}
%\newblock \emph{{The American Heritage. Dictionary of the American Language}}.
%\newblock Houghton Mifflin Company, 1992.

%\end{thebibliography}

% \appendix
% \section{A summary of Latin grammar}    % Each appendix must have a short title.
% \section{Some Latin vocabulary}              % Sections and subsections are supported  
                                                                         % in the appendices.
\end{document}